\documentclass[]{DAC}
\title{Exploration of Always $S$-Connected Temporal Graphs}

\author{
  Duncan Adamson\\
  \emph{University of St Andrews}\\
  \texttt{duncan.adamson@st-andrews.ac.uk}
  \and
  Paul G. Spirakis\\
  \emph{University of Liverpool}\\
  \texttt{p.spirakis@liverpool.ac.uk}
}
\date{}

\newcommand{\TEXP}{\textsc{TEXP}}

\begin{document}

\maketitle

\begin{abstract}
\emph{Temporal graphs} are a generalisation of (static) graphs, defined by a sequence of \emph{snapshots}, each a static graph defined over a common set of vertices. \emph{Exploration} problems are one of the most fundamental and most heavily studied problems on temporal graphs, asking if a set of $m$ agents can visit every vertex in the graph, with each agent only allowed to traverse a single edge per snapshot. In this paper, we introduce and study \emph{always $S$-connected} temporal graphs, a generalisation of always-connected temporal graphs where, rather than forming a single connected component in each snapshot, we have at most $\vert S \vert$ components, each defined by the connection to at least one vertex in the set $S$. We use the model of always $S$-connected temporal graphs to study \emph{subgraphs} of always-connected temporal graphs, with motivation coming from networks with a small number of ``choke-points'', such as rail networks, where all external traffic into a subnetwork comes from a limited number of external connections.

We use this formulation as a tool for exploring always-connected temporal graphs admitting an \emph{$(r,b)$-division}, a partitioning of the vertex set into $r$ components, each of which is $S$-connected, where $\vert S \vert \leq b$. We show that an always $S$-connected temporal graph with $m = \vert S \vert$ and an average degree of $\Delta$ can be explored by $m$ agents in $O(n^{3/2} m^2 \Delta^{3/2}\log^{3/2}(n) +  + nm^3)$ snapshots. Using this result as a subroutine, we show that any always-connected temporal graph with treewidth at most $k$ can be explored by a single agent in $O\left(n^{4/3} k^{9/2}\log^{5/2}(n)\right)$ snapshots, improving on the current state-of-the-art for small values of $k$.

Finally, we provide new results on the exploration of \emph{simple temporal graphs}, temporal graphs where each edge is active in at most one snapshot. We show that the problem of determining if it is possible to explore a given simple temporal graph is NP-Hard, while showing that if the underlying graph contains $k$ chords, then we can determine if it can be explored by a single agent in $O(n + (k^2 + k)!)$ time.
\end{abstract}

\section{Introduction}

In many real-world settings, networks are not static objects but instead have unstable connections that vary over time. Such examples include public transport networks \cite{Kutner2025a} and infection control \cite{ruget2021multi}. Temporal graphs provide a model for such time-varying networks. Formally, a \emph{temporal graph $\mathcal{G}$} is a generalisation of (static) graphs containing a common vertex set $V$ and ordered sequence of $T$ \emph{snapshots} $G_1, G_2, \dots, G_T$, with the \emph{lifetime} of the graph defined by the number of snapshots, by convention denoted $T$. Unlike some models of dynamic graphs, we assume that we have full knowledge of the graph, i.e. that we are given every snapshot as part of the input.

Much work on temporal graphs has focused on the problem of the \emph{exploration} of temporal graphs. In the \emph{temporal exploration problem}, denoted \TEXP, every vertex must be visited at least once by one of a set of agents. The movement of the agents is restricted by allowing only a single move per snapshot. Thus, as the graph changes, the agents' movement may become restricted. We assume our objective is to find a set of \emph{temporal walks}, a generalisation of walks accounting for the movement restrictions imposed by only traversing a single edge per snapshot, covering the vertex set, with exactly one walk per agent. We call such a set of walks an \emph{exploration schedule}.

Work on \TEXP~was motivated by extending the well-known \textsc{Travelling Salesman} to temporal graphs by Michail and Spirakis \cite{michail2016traveling}. Since then, there has been a broad span of work covering reachability \cite{DeligkasP22,Enright2021,meeks2022reducing} and exploration \cite{Adamson2022,Arrighi2023,Baguley2025,dogeas2023exploiting,erlebach2021temporal,erlebach2019two,erlebach2022exploration,erlebach2023parameterised,Kutner2025a,michail2016traveling}.
The decision version of \textsc{TEXP} in which one has to decide if at least one exploration schedule exists in a given temporal graph from a given starting vertex is an \textbf{NP}-complete problem \cite{michail2016traveling}. Indeed, this problem remains \textbf{NP}-complete even if the underlying graph has pathwidth 2 and every snapshot is a tree \cite{bodlaender2019exploring}, or if the underlying graph is a star and the exploration has to start and end at the centre of the star \cite{Akrida2021}.
In the optimisation version of this problem, we ask for the \emph{quickest} exploration of a given temporal graph, that being the schedule ending at the earliest snapshot, formally defined in Section \ref{sec:prelims}.\looseness=-1

On the positive side, we mention the landmark paper by Erlebach et al. \cite{erlebach2021temporal}, who show that any \emph{always-connected} temporal graph (temporal graphs where each snapshot is connected) can be explored in $O(n^2)$ snapshots, where $n$ is the number of vertices. This has recently been improved by Bastide et al. \cite{Bastide2025}, who have shown that the always-connected temporal graphs with average degree $\Delta$ can be explored in $O(n^{3/2} \sqrt{\Delta \log n})$ snapshots, simplifying to $O(n^{3/2} \sqrt{\log n})$ for graphs of constant average degree, such as planar graphs, and to $O(n^{3/2} \sqrt{k \log n})$ for graphs of treewidth $k$. At the same time, Baguley et al. \cite{Baguley2025} have shown that if each snapshot is a random spanning tree of the graph, then, regardless of the degree, the graph can, with high probability, be explored in the lower of $O(n^{3/2})$ or $O(m)$ snapshots, where $n$ is the number of vertices and $m$ the number of edges.\looseness=-1

Moving to other classes of always connected temporal graphs, in \cite{erlebach2021temporal} Erlebach et al. showed that if the underlying graph is a $2 \times n$ grid, then the temporal graph can be explored in $O(n \log^3 n)$ snapshots. Erlebach and Spooner \cite{erlebach2022exploration} showed that \emph{$k$-edge-deficient} temporal graphs, temporal graphs in which each snapshot has at most $k$-edges removed from the underlying graph, can be explored in $O(k n \log n)$ snapshots, or $O(n)$ when only a single edge is removed per snapshot. Finally, if the underlying graph is a cycle or a cycle with a single chord, then the temporal graph can be explored in $O(n)$ snapshots, later generalised by Alamouti \cite{taghian2020exploring} to exploring temporal graphs with an underlying graph as a cycle with $k$-chords, giving an $O(k^2 k! e^k n)$ upper bound on exploration, strengthened by Adamson et al. \cite{Adamson2022} to an upper bound of $O(kn)$.\looseness=-1

On the negative side, we have, by Erlebach et al. \cite{erlebach2021temporal}, that there exists a class of always-connected temporal graphs requiring $\Omega(n^2)$ snapshots to explore, and that there exists temporal graphs with an underlying planar graph of degree at most $4$ that cannot be explored faster than in $\Omega(n \log n)$ snapshots. Further, in \cite{erlebach2022exploration}, Erlebach and Spooner showed that there exist some $k$-edge-deficient temporal graphs that cannot be explored faster than in $\Omega(n \log k)$ snapshots.

\subsection{Our Results}

The primary result of this paper is a generalisation of the work by Bastide et al. \cite{Bastide2025} for always-connected temporal graphs to \emph{$S$-connected temporal graphs}, temporal graphs such that there exists a set $S$ where every vertex is connected to at least one vertex in $S$ in each snapshot. In doing so, we use our generalisation as a tool to obtain faster exploration schedules for graphs of bounded treewidth, even when restricted to a single agent. Our primary result is the following theorem:

\begin{theorem}
    \label{thm:exploring_S_always_connected}
    Let $\mathcal{G} = G_1, G_2, \dots, G_T$ be an always $S$-connected temporal graph with lifetime $$T = \Omega\left(
    n m^2 \Delta^{3/2} \sqrt{\vert X \vert \log\left(m\vert X \vert\right)}\log(\vert X \vert) + nm^3\right)$$ where $m = \vert S \vert$ and $\Delta$ is the average degree of the underlying graph. Then any subset $X \subseteq V$ of vertices can be explored by $m$ agents.
\end{theorem}

\newtheorem*{theorem_exploring_S_always_connected}{Theorem \ref{thm:exploring_S_always_connected}}

This is applied to get the following theorem for the exploration of always-connected temporal graphs with underlying graphs of bounded treewidth.

\begin{theorem}
    \label{thm:treewidth_exploration}
    Any always-connected temporal graph whose underlying graph has treewidth $k$ can be explored in $O\left(n^{4/3} k^{9/2}\log^{5/2}(n)\right)$ snapshots.
\end{theorem}

\newtheorem*{theorem_treewidth_exploration}{Theorem \ref{thm:treewidth_exploration}}

Note that for graphs of small or constant treewidth, this is a strict improvement. Further, we hope that our tools will allow for more classes of graphs to be efficiently explored via partitioning approaches such as those used in \cite{Adamson2022} and \cite{erlebach2021temporal} for the exploration of graphs of bounded treewidth.

Finally, we look at the problem of exploring \emph{simple} temporal graphs, temporal graphs where each edge is only present in a single snapshot. We show that the problem of determining if a given simple temporal graph can be explored is NP hard and that we can determine in $O(n + (k^2 + k)!)$ time if a single agent can explore a simple temporal graph whose underlying graph is a cycle with at most $k$ chords.

\section{Preliminaries}
\label{sec:prelims}

We first define the notation used in this paper. Let $[i, j] = i, i + 1, \dots, j$ denote the (ordered) set of integers between some pair $i, j \in \mathbb{N}$, where $\mathbb{N}$ is the set of natural numbers. Note that if $i = j$, then $[i, j] = \{i \}$, and if $i > j$, $[i, j] = \emptyset$. For notational conciseness, we us $[i]$ as shorthand for $[1, i]$.

We define a \emph{graph} $G = (V, E)$ by a set of vertices, by convention $V = \{v_1, v_2, \dots, v_n\}$, and a set of edges, by convention $E \subseteq V \times V$, each a tuple of vertices. Note that we may write the edge $e$ between the vertices $v_i$ and $v_j$ as either $e = (v_i, v_j)$ or $e = (v_j, v_i)$. When an edge is given explicitly as $(v_i, v_j)$, we call $v_i$ the \emph{start point} and $v_j$ the \emph{end point}. A \emph{walk} in a graph is an ordered sequence of edges $W = (v_{i_1}, v_{i_2}), (v_{i_2}, v_{i_3}), \dots, (v_{i_{m - 1}, v_{i_m}})$, where the end point of the $j^{th}$ edge is the start point of the $(j + 1)^{th}$ edge. The start point of the first edge in a given walk is the start point of the walk and, analogously, the end point of the walk is the end point of the last edge in the walk. The \emph{length} of a walk $W$, denoted $\vert W \vert$ is the number of edges in the walk. A graph $G = (V, E)$ is \emph{connected} if there exists, for every pair of vertices $v, u \in V$, a walk starting at $v$ and ending at $u$. Note that a single walk $W$ may contain multiple copies of the same edge without contradiction. On the other hand, a \emph{path} is a walk where no vertex is in more than two edges, and both the start and end vertices belong to a single edge, the first and last respectively.

The \emph{distance} between two vertices $v, u \in V$ in a graph, denoted $\dist(v,u)$ is the value such that there exists some walk $W$ of length $\dist(v, u) $ with the start point $v$, the end point $u$ and $\forall W' \in \{W''$ $\mid$ $W''$ is a walk in $G$ starting at $v$ and ending at $u \}$, $\vert W' \vert \geq \dist(v, u)$. Note that each shortest walk between any pair of vertices will be a path.
The \emph{neighbourhood} of a vertex $v$, denoted $N(v)$, is the set of vertices sharing an edge with $v$, formally $N(v) = \{u \in V \mid (v, u) \in E \}$. The \emph{degree} of a vertex $v$, denoted $\Delta(v)$ is the number of vertices in the neighbourhood of $v$, formally, $\Delta(v) = \vert N(v) \vert$.

A \emph{temporal graph} is a generalisation of a graph, herein called a \emph{static graph} whenever confusion may otherwise arise, where, rather than having a single edge set, the graph contains an ordered sequence of \emph{snapshots}, by convention $G_1, G_2, \dots, G_T$, each of which is a static graph over a shared set of vertices. We assume that $G_t = (V, E_t)$, for every $t \in [1, T]$. An edge $e$ is \emph{active} at snapshot $t$ if $e \in E_t$. In our definition, an edge $e$ may be active in any number of snapshots. The \emph{temporal neighbourhood} of a vertex $v \in V$ at snapshot $G_t$, denoted $N_t(v)$, is equal to the neighbourhood of $v$ in the static graph $G_t$.


The \emph{underlying graph} of a temporal graph $\mathcal{G} = G_1, G_2, \dots, G_T$, denoted $U(\mathcal{G})$, is the static graph $U(\mathcal{G}) = (V, \bigcup_{t \in [1, T]} E_t)$, i.e. the static graph formed with the edge set corresponding to the union of all snapshots in the graph. For simplicity, given a vertex $v$ in a temporal graph $\mathcal{G}$ we use $\Delta(v)$ to denote the neighbourhood of $v$ in the underlying graph $U(\mathcal{G})$. In general, given some property $X$ of a static graph, a temporal graph has the property $X$ if the underlying graph satisfies the property. A temporal graph \emph{always satisfies} a given property $X$ if each snapshot $G_1, G_2, \dots, G_T$ also satisfies $X$. Notably, a static graph $G = (V, E)$ is \emph{connected} if, for each pair of vertices $v_i, v_j \in V$, there exists some walk starting at $v_i$ and ending at $v_j$. Therefore, a temporal graph $\mathcal{G} = G_1, G_2, \dots, G_T$ is an \emph{always-connected temporal graph} if the static graph $G_t$ is connected, for every $t \in [T]$.

In this paper, we are interested in \emph{always $S$-connected temporal graphs}, a generalisation of always-connected temporal graphs.

\begin{definition}[Always $S$-connected temporal graphs]
    A temporal graph $\mathcal{G} = G_1, G_2, \dots, G_T$ is \emph{Always $S$-Connected}, for a given subset of vertices $S \subseteq V$, if, for every vertex $v \in V$ and snapshot $G_t$, there exists at least one vertex $u \in S$ such that there is a walk from $u$ to $v$ in $G_t$.
\end{definition}

Note that an always-connected temporal graph is an always $S$-connected temporal graph for any non-empty subset $S$ of the vertex set $V$.

A temporal graph is \emph{simple} if each edge is active in exactly one snapshot, equivalently, $\forall t, t' \in [1, T]$ where $t \neq t'$, $E_t \cap E_{t'} = \emptyset$. A temporal graph is proper if no vertex has more than a single incident edge per snapshot, equivalently, $\forall v \in V$, $t \in [1, T]$, $\Delta_t(v) \leq 1$.

\subsection{Temporal Walks and Exploration}

A \emph{temporal walk} is a set of edge-snapshot index tuples, $\mathcal{W} = ((v_{i_1}, v_{i_2}), t_1)$, $((v_{i_2}, v_{i_3}), t_2)$, $\dots$, $((v_{i_{m - 1}}, v_{i_m}), t_{m - 1})$ such that:
\begin{itemize}
    \item $(v_{i_1, i_2}), (v_{i_2}, v_{i_3}), \dots, (v_{i_{m - 1}, v_{i_m}})$ form a walk in the underlying graph $U(\mathcal{G})$,
    \item the edge $(v_{i_j}, v_{i_{j + 1}})$ is active in snapshot $t_j$, and,
    \item $1 \leq t_1 < t_2 < \dots < t_{m - 1} \leq T$.
\end{itemize}

\noindent
Note that this definition means that a temporal walk cannot include multiple edges from the same snapshots, but may ``wait'' at a given vertex between edge transitions. The \emph{length} of a temporal walk $\mathcal{W}$, denoted $\vert \mathcal{W} \vert$ is the snapshot $t_{m - 1}$ associated with the final edge. We refer to the $i^{th}$ tuple of a temporal walk as the $i^{th}$ \emph{step} in the walk.

A temporal walk \emph{explores} a static graph $G$ if, for every vertex $v \in V$, there exists an edge containing $v$ in at least one edge of the temporal walk, or, in the case of multiple agents, if some agent starts on a vertex and does not move. Analogously, a set of temporal walks explores a temporal graph $\mathcal{G}$ if, for every vertex $v \in V$, there exists an edge containing $v$ in at least one edge of one of the walks. An \emph{exploration schedule} is a set of temporal walks that explores a temporal graph.

\begin{problem}[Temporal graph exploration problem]
    \label{prob:exploration_problem_decide}
    Given a temporal graph, $\mathcal{G}$ $=$ $G_1$, $G_2$, $\dots$, $G_T$, and a $k$-tuple of start vertices from $V$, $V' = (v_1, v_2, \dots, v_k)$, does there exist a set of $k$ temporal walks $\mathcal{W} = \{W_1, W_2, \dots, W_k\}$ that explores $\mathcal{G}$ such that $W_i$ starts at $v_i$?
\end{problem}

Note that under the above definition, there is no requirement that the vertices in $V'$ are unique, i.e., there may exist some pair $i, j \in [k]$ where $i \neq j$ and $v_i = v_j$.
In general, we assume that an exploration is undertaken by a set of $k$-\emph{agents}, each starting at some vertex in $V'$, and matched uniquely to some temporal walk. When the temporal graph is always $S$-connected, we assume that we have (at least) one agent starting on each vertex in $S$. We do so to ensure that the graph can be explored, as otherwise there exist trivial examples of temporal graphs which can not be explored, e.g. when some vertex in $S$ is never connected to any other, thus being trivially unreachable. We note that this problem is NP-complete in many, even very simple, cases \cite{bodlaender2019exploring,erlebach2021temporal}. Therefore, we focus in this paper on finding upper bounds on the number of snapshots required to guarantee that a temporal graph can be explored.

\subsection{Tools}

In this section, we briefly note some of the key tools used in this paper in deriving our main results. We note that we have slightly adapted the language of some of these lemmas to bring them in line with the terminology used within the remainder of this paper, without altering the meaning.

\begin{lemma}[Reachability, Lemma 2.1 in \cite{erlebach2021temporal}]
    \label{lem:connection_lemma}
    Let $\mathcal{G}$ be a temporal graph with a vertex set $V$. Given a pair of vertices $v, u \in V$ such that there exists a sequence of $\vert V \vert$ snapshots $G_{i_1}, G_{i_2}, \dots, G_{i_{\vert V \vert}}$ where $v$ and $u$ are connected in each snapshot, and $i_1 < i_2 < \dots < i_{\vert V \vert}$, then there exists a temporal walk from $v$ to $u$ in the temporal graph $G_{i_1}, G_{i_2}, \dots, G_{i_{\vert V \vert}}$.
\end{lemma}

\begin{lemma}[Multi to single agent exploration, Lemma 2.2 in \cite{erlebach2021temporal}]
    \label{lem:multi_agent_to_single_agent}
    Given a class of always-connected temporal graphs that can be explored by $k$ agents in $O(T)$ snapshots, the same class can be explored by a single agent in $O((T + n)k\log n)$ snapshots.
\end{lemma}

\subsection{Lower bounds on Exploring Always $S$-Connected Graphs}

We note first that any lower bounds for always-connected temporal graphs automatically apply to always $S$-connected temporal graphs. More specifically, assuming we are exploring with $\vert S \vert$-agents we have from \cite{erlebach2021temporal}:
\begin{itemize}
    \item that there exists an always $S$-connected temporal graph requiring $\Omega((n + 1- \vert S \vert)^2)$ snapshots to explore, and,
    \item for any degree $D$, there exists an always $S$-connected temporal graph with maximum degree $D$ requiring $\Omega(D (n + 1 - \vert S \vert))$ snapshots to explore.
\end{itemize}

\noindent
Note that both bounds are reached by using the constructions found in \cite{erlebach2021temporal} to find an always-connected temporal graph with $n + 1 - \vert S \vert$ vertices requiring $\Omega((n + 1- \vert S \vert)^2)$ or $\Omega(D (n + 1 - \vert S \vert))$ snapshots, then appending an additional $\vert S \vert - 1$ vertices, none of which are connected to any other vertex during the lifetime of the graph.

\section{\texorpdfstring{Exploring Always $S$-Connected Graphs}{Exploring S-always-connected Graphs}}
\label{sec:big_generalisations}

In this section, we provide our primary result, namely that any subset of vertices $X \subseteq V$ in an always $S$-connected graph can be explored by $m = \vert S \vert$ agents in $O\left(n m^3 \Delta^{3/2} \sqrt{\vert X \vert \log\left(m\vert X \vert\right)}\log(\vert X \vert)\right)$
snapshots, where $\Delta$ is the average degree of any vertex in the underlying graph, i.e. $\Delta = \sum_{v \in V} \Delta(v) / \vert V \vert$. We note this directly generalises the results from \cite{Bastide2025}, with the key difference being the generalisation from always-connected temporal graphs to always $S$-connected temporal graphs.
We show, in Section \ref{sec:applications}, that this generalisation allows for faster exploration of several key classes of graphs, in particular those with bounded treewidth. 

We assume, for the remainder of this section, that $S$ contains $m$ vertices, that the temporal graph contains $n$ vertices, and that the average degree of the underlying graph is $\Delta$. We will restate this notation within the key theorems of this section.
We begin with a small Lemma that provides an immediate tool for exploring such graphs in $O(n^2 m)$ snapshots.

\begin{lemma}
    \label{lem:there_and_back_again}
    Let $\mathcal{G} = G_1, G_2, \dots, G_T$ be an always $S$-connected temporal graph. Then, given a set of $m$ agents positioned on the vertices of $S$ and some vertex $v \in V$, there is at least one agent that can reach $v$ and return to its starting location in $O(n m)$ snapshots.
\end{lemma}

\begin{proof}
    Observe that in each snapshot, there exists some path of length at most $n$ from $v$ to the starting location of at least one agent. As there are $m$ possible starting positions, over any set of $2nm$ snapshots, there must exist some starting vertex $u \in S$ that has a walk of length at most $n$ to $v$ in $2n$ snapshots, and thus, by Lemma \ref{lem:connection_lemma}, the agent starting at $u$ can visit $v$ then return to $u$ in $2n m$ snapshots.
\end{proof}

Note that Lemma \ref{lem:there_and_back_again} allows for an exploration of any always $S$-connected temporal graph in $O(n^2 m)$ snapshots by exploring each vertex, one after another, each requiring $O(n m)$ snapshots, with the key property that, after exploring a given vertex, the agent returns to its starting vertex in $S$, allowing the repeated application of the lemma. We now provide our key combinatorial tools for lowering the upper bound on the number of snapshots needed to explore the graph, based on those provided in \cite{Bastide2025}, for exploring always $S$-connected temporal graphs with a fixed number of components defined by a subset of vertices.

\begin{lemma}
    \label{lem:subset_point_to_point}
    Let $\mathcal{G} = G_1, G_2, \dots, G_T$ be an always $S$-connected temporal graph with average degree $\Delta$ and $T \geq 2\Delta n / (\vert X \vert - (m - 1)) + 1$. Then, given any subset of vertices $X \subseteq V$ where $\vert X \vert \geq m + 1$, there exists at least one pair of vertices $v, u \in X$ such that there is a temporal walk starting at $v$ and ending at $u$ within the lifetime of the graph.
\end{lemma}

\begin{proof}
    For brevity, let $S = \{s_1, s_2,\dots, s_m\}$, and let $C_{i, t}$ denote the set of vertices in $G_t$ connected to $s_{i}$. Note that in each snapshot either $C_{i, t} = C_{j, t}$, if $s_i$ is connected to $s_j$, or $C_{i, t} \cap C_{j, t} = \emptyset$ otherwise. We introduce the following sets:
    \begin{itemize}
        \item $F_{i,t}(v)$ (for \textbf{F}orwards) as the set of vertices in component $C_{i, t}$ that can be reached from vertex $v$ by a temporal walk starting in the first snapshot. Formally, $F_{i, t}(v) = \{u \in C_{i,t} \mid \exists$ a temporal walk in $G_1, G_2, \dots,G_t$ starting at $v$ and ending at $u\}$.
        \item $B_{i,t}(v)$ (for \textbf{B}ackwards) as the set of vertices in component $C_{i,t}$ that can reach $v$ by temporal walk stating at each vertex in snapshot $t$ by snapshot $T$. Formally, $B_{i, t}(v) = \{u \in C_{i,t} \mid \exists$ a temporal walk in $G_t, G_{t + 1}, \dots,G_T$ starting at $u$ and ending at $v\}$.
    \end{itemize}

    \noindent
    As a base case, $F_{{i,1}}(v) = \left(\{v\} \cup N_1(v)\right) \cap C_{i,1}$ and $B_{{i, T}}(v) =\left( \{v\}\cup N_T(v)\right) \cap C_{i, T}$. Informally, the only vertices that can be reached in the first snapshot from $v$ are those in the neighbourhood of $v$ in that snapshot, alongside $v$ itself. Equivalently, $v$ can only be reached by snapshot $T$ in snapshot $G_T$ by those vertices in the neighbourhood of $v$ in $G_T$, alongside $v$ itself.

    We now consider the size of the union of the sets $F_{{i,t}}(v)$ over the set of vertices $V$. First, observe that the vertex $u$ is added to $F_{{i,t}}(v)$ if $u \in C_{i,t}$ and at least one of the following holds:
    \begin{itemize}
        \item $u \in \bigcup_{j \in [m]} F_{{j, t - 1}}(v)$, or,
        \item $\exists u' \in \bigcup_{j \in [m]} F_{{j, t - 1}}(v)$ such that $(u', u) \in E_t$.
    \end{itemize}

    \noindent
    Observe that, as any vertex $u$ that can be reached by snapshot $t - 1$ by a temporal walk from $v$ starting at snapshot $1$ can also be reached by a temporal walk from $v$ starting at snapshot $1$ by snapshot $t$, $\bigcup_{i \in [m]} F_{{i,t - 1}}(v) \subseteq \bigcup_{i \in [m]} F_{{i,t}}(v)$ and thus $\left\vert \bigcup_{i \in [m]} F_{{i,t}}(v)\right\vert \geq \left\vert \bigcup_{i \in [m]} F_{{i,t - 1}}(v)\right\vert$. Further, $\bigcup_{i \in [m]} F_{{i,t - 1}}(v)$ is equal to $\bigcup_{i \in [m]} F_{{i,t}}(v)$ if and only if, for every component $C_{i, t}$, either $C_{i, t} = F_{{i,t}}(v)$ or $ F_{{i,t}}(v) = \emptyset$ as otherwise there exists some vertex $u \in F_{{i,t}}(v)$ connected to some $u' \in C_{i,t} \setminus F_{{i,t}}(v)$, and thus, $u'$ can be added to $\bigcup_{i \in [m]} F_{{i,t}}(v)$ in snapshot $t$, contradicting the assumption that $\bigcup_{i \in [m]} F_{{i,t - 1}}(v) = \bigcup_{i \in [m]} F_{{i,t}}(v)$. Therefore, if by snapshot $t$ there does not exist any pair $v, u \in X$ such that there exists a walk from $v$ to $u$, then the maximum number of vertices in the set $X' = \{ v \in X \mid \left\vert \bigcup_{i \in [m]} F_{{i, (t - 1)}}(v) \right\vert = \left\vert \bigcup_{i \in [m]} F_{{i,t}}(v) \right\vert\}$ is $m - 1$, corresponding to the case where at most one vertex is in $C_{i,t}$ for $m - 1$ components, and the remainder are in a single common component.
    Hence, $\sum_{v \in X} \left\vert\bigcup_{i \in [m]} F_{{i, t}}(v) \right\vert \geq \sum_{v\in X}\left\vert\bigcup_{i \in [m]} F_{{i, t - 1}}(v) \right\vert + \vert X \vert - (m - 1)$, for any snapshot $G_t$ such that there is no temporal walk between any pair of vertices in $X$ in the temporal graph $G_1, G_2, \dots, G_t$.

    In the other direction, observe that the vertex $u$ is added to $B_{{i,t}}(v)$ if $u \in C_{i,t}$ and at least one of the following hold:
    \begin{itemize}
        \item $u \in \bigcup_{i \in [m]} B_{{j, t + 1}}(v)$, or,
        \item $\exists u' \in \bigcup_{i \in [m]} B_{{i, t + 1}}(v)$ such that $(u', u) \in E_t$.
    \end{itemize}

    \noindent
    As with the sets $F_{{i, t}}(v)$, we can place bounds on the size of $\bigcup_{i \in [m]} B_{{i,t}}(v)$ relative to that of $\bigcup_{i \in [m]} B_{{i,t + 1}}(v)$. First, note that $\bigcup_{i \in [m]} B_{{i,t + 1}}(v) \subseteq \bigcup_{i \in [m]} B_{{i,t}}(v)$. Further, if $\bigcup_{i \in [m]} B_{{i,t}}(v) = \bigcup_{i \in [m]} B_{{i,t + 1}}(v)$ then for every $i \in [m]$, either $C_{i, t} = B_{{i,t}}(v)$ or $C_{i, t} \cap B_{{i,t}}(v) = \emptyset$. Hence, following the same arguments as in the forward direction, there can be at most $m - 1$ vertices in $X$ for which $\left\vert \bigcup_{i \in [m]} B_{{i,t}}(v)\right\vert = \left\vert \bigcup_{i \in [m]} B_{{i,t + 1}}(v)\right\vert$, and thus $\sum_{v \in X} \left\vert \bigcup_{i \in [m]} B_{{i,t}}(v)\right\vert \geq \sum_{v \in X}\left\vert \bigcup_{i \in [m]} B_{{i,t + 1}}(v)\right\vert + \vert X \vert - (m - 1)$.

    We can now prove the statement of the lemma.

    Consider some vertex $v \in X$ where $\left\vert \bigcup_{i \in [m]} F_{{i,t}}(v) \right\vert > \left\vert \bigcup_{i \in [m]} F_{{i,t - 1}}(v)\right\vert$ and $\left\vert \bigcup_{i \in [m]} B_{{i,t}}(v) \right\vert$ > $\left\vert \bigcup_{i \in [m]} B_{{i,t + 1}}(v)\right\vert$, noting that if $\vert X \vert \geq m + 1$ at least one such vertex must exist by pigeon hole argument, and let  $I_t(v) =$ \\$N_{G_t}\left(\left(\bigcup_{i \in [m]} F_{{i,t - 1}}(v) \right) \cap \left(\bigcup_{i \in [m]} B_{{i,t + 1}}(v)\right)\right)$ be the set of vertices in\\ $V \setminus \left(\left(\bigcup_{i \in [m]} F_{{i,t - 1}}(v) \right) \cap \left(\bigcup_{i \in [m]} B_{{i,t + 1}}(v)\right)\right)$ that are adjacent to at least one vertex in $\left(\bigcup_{i \in [m]} F_{{i,t - 1}}(v) \right) \cap \left(\bigcup_{i \in [m]} B_{{i,t + 1}}(v)\right)$. Note that $I_t(v)$ is empty only if either $C_{i, t} \subseteq \left(\bigcup_{i \in [m]} F_{{i,t - 1}}(v) \right) \cap \left(\bigcup_{i \in [m]} B_{{i,t + 1}}(v)\right)$ or $C_{i, t} \cap \left(\bigcup_{i \in [m]} F_{{i,t - 1}}(v) \right) \cap \left(\bigcup_{i \in [m]} B_{{i,t + 1}}(v)\right) = \emptyset$, for every $i \in [m]$. As above, if there exists any $u \in X$ such that $u \in C_{i, t}$ for any $C_{i, t}$ where $C_{i, t} \subseteq \left(\bigcup_{i \in [m]} F_{{i,t - 1}}(v) \right) \cap$ \\ $\left(\bigcup_{i \in [m]} B_{{i,t + 1}}(v)\right)$, then we have a walk from $v$ to $u$, within $T$ snapshots. Similarly, if $C_{i, t} \subseteq \left(\bigcup_{i \in [m]} F_{{i,t - 1}}(v) \right) \cap \left(\bigcup_{i \in [m]} B_{{i,t + 1}}(v)\right)$ and $C_{i, t} \subseteq$\\ $\left(\bigcup_{i \in [m]} F_{{i,t - 1}}(v') \right) \cap \left(\bigcup_{i \in [m]} B_{{i,t + 1}}(v')\right)$, then there is a temporal walk from $v$ to $v'$ via some vertex in $C_{i,t}$. Therefore, there can be at most $m - 1$ vertices in $X$ for which $I_t(v)$ is empty. Otherwise, we have at least one vertex $w \in I_t(v)$, which we call the \emph{recorded} vertex.\looseness=-1

    Following the same argument as in \cite{Bastide2025}, there exists, for any $v \in X, w \in I_t(v)$, at most two snapshots, $t_1, t_2$, such that $w \in I_{t_1}(v)$ and $w \in I_{t_2}(v)$, noting that, if any third snapshot $t_3$ exists where $w \in I_{t_3}(v)$ and $t_1 < t_2 < t_3$, then $v$ can already be reached from $w$ starting at snapshot $t_3$, and $w$ can be reached from $v$ by snapshot $t_1$. Thus, it can not be in the neighbourhood of either $\bigcup_{i \in [m]} F_{{i, t_2 - 1}}(v)$, nor $\bigcup_{i \in [m]} B_{{i, t_3 - 1}}(v)$. Hence, $w$ can be in $I_t(v)$ for at most two snapshots for each $w$. More generally, if $w$ is recorded more than $2 \Delta(w)$ times, corresponding to the situation where every neighbour of $w$ can reach $w$ only once on some snapshot $t$, then there must exist some pair of snapshots $t_1, t_2$ where $t_1 < t_2$ and vertices $v_1, v_2 \in X$ such that $w \in I_{t_1}(v_1)$ and $w \in I_{t_2}(v_2)$, and hence there exists a temporal walk from either $v_1$ to $v_2$ or $v_2$ to $v_1$. Therefore, as $\sum_{v \in X} \vert I_t(v)\vert \geq \vert X \vert - (m - 1)$ holds for all $r \in [T]$, by snapshot $T$, we have recorded at least $2 \Delta n + \vert X \vert - (m - 1)$ vertices, and thus there must exist some pair $v, u \in X$ such that there exists a temporal walk from $v$ to $u$.
\end{proof}

\begin{lemma}
    \label{lem:small_set_big_impact}
    Let $\mathcal{G} = G_1, G_2, \dots, G_T$ be an always $S$-connected temporal graph with average degree $\Delta$ where $T = 2\Delta n/k + 1$, for some $k \geq 2$. Then, given any subset of vertices $X \subseteq V$ where $\vert X \vert \geq 2m + 2$, there exists a subset $R \subseteq X$ of size at most $4k \log(\vert X \vert)$ such that, for every $v \in X$ there exists some $u \in R$ such that there exists a walk from $u$ to $v$.
\end{lemma}

\begin{proof}
    We adapt the arguments in \cite{Bastide2025} to the multiple components case.

    First, if $\vert X \vert \leq 2k \log (\vert X \vert)$, then this holds directly. Otherwise, given some set of vertices $\chi \subseteq X$, let $F_{\chi}(v) = \{u \in \chi \setminus \{v\}\mid$ there exists a temporal walk in $\mathcal{G}$ from $v$ to $u\}$, and $B_{\chi}(v) = \{u \in \chi \setminus \{v\}\mid$ there exists a temporal walk in $\mathcal{G}$ from $u$ to $v\}$. Now, we construct a sequence of sets $X_1, X_2, \dots, X_{\ell}$ and vertices $v_1, v_2, \dots, v_{\ell}$ such that, $X_1 = X$, $X_i = X_{i - 1} \setminus \left(F_{X}(v_{i - 1}) \cup B_{X}(v_{i - 1})\right)$, and $X_\ell \subseteq F(v_{\ell})$. At each step, we choose $v_i$ from the set $X_{i}$ as some vertex where $\vert F_{X_i}(v_i) \vert \geq \vert B_{X_i}(v_i) \vert$, and add $v_i$ to the set $R$.
    
    Now, note that by the above definition, for any $v_i, v_j$ where $i \neq j$, there can be no temporal walk from $v_i$ to $v_j$ or $v_j$ to $v_i$ in $\mathcal{G}$. We assume that $\ell$ is minimal such that this holds. Following the same arguments as Lemma \ref{lem:subset_point_to_point}, we have that in $2\Delta n/k$ snapshots, at least $2\Delta n (\vert X_i \vert - (m- 1)) / k + 1$ vertices will be recorded by the vertices in $X_i$, and, therefore, there must be at least one vertex that records at least $(\vert X_i \vert - (m - 1))/ k \geq \vert X_i \vert/2k$ other vertices in $X_i$. Hence,

    \[
        \vert F_{X_i}(v_i)\vert \geq \vert F_{X_{i}}(v_i) \cup B_{X_{i}}(v_i)\vert/2 \geq \left((\vert X_i \vert - (m - 1))/ 2k \right) - 1 \geq \vert X_i \vert/4k
    \]
    
    Therefore, after adding $4k$ vertices to $R$ we have $\vert X_{4k} \vert \leq (\vert X \vert)/4 + (\vert X \vert / 4k) \leq \vert X \vert / 2$.
    %
    Repeating this at most $\log(\vert X \vert)$ times gives the result on the size of $R$. 
\end{proof}

\begin{lemma}
    \label{lem:path_length}
    Let $\mathcal{G} = G_1, G_2, \dots, G_T$ be an always $S$-connected temporal graph with average degree $\Delta$ where $T \geq 2n$ and $\vert X \vert \geq 4m^2 + 4m + 4$. Then, there exists a temporal walk in $\mathcal{G}$ that covers at least $\sqrt{\vert X \vert / (\Delta \log \vert X \vert)}/32$ vertices in any subset $X \subseteq V$.
\end{lemma}

\begin{proof}
    We partition $\mathcal{G}$ into a set of $T / P$ \emph{epochs}, $\mathcal{G}^1, \mathcal{G}^2, \dots, \mathcal{G}^{T / P}$, each of length $P = 1 + 32n/\left\lfloor\sqrt{\vert X \vert/(\Delta \log \vert X \vert)} \right\rfloor$, where $\mathcal{G}_i = G_{(i - 1)P + 1}, G_{(i - 1)P + 2}, \dots, G_{iP}$.
    Further, we construct a sequence of subsets $X_1, X_2, \dots, X_{T /P}$ each of size at most
    $\sqrt{\vert X \vert \Delta \log \vert X \vert} / 4$, where $X_i \subseteq X \setminus \left(\bigcup_{j \in [1, i - 1]} X_{j} \right)$ satisfies that, for every vertex $u \in X \setminus \bigcup_{j \in [i - 1]} X_j$, there is a vertex $v \in X_i$ with a temporal walk from $v$ to $u$ in $\mathcal{G}^{i}$.
    
    To prove this is possible, for the first set $X_1$, we can, using Lemma \ref{lem:small_set_big_impact}, find a set of size $4k \log \vert X \vert$ where $2 \Delta n / k + 1 = P$, satisfied by setting $k = \sqrt{\vert X \vert \Delta} / 16 \sqrt{\log \vert X \vert}$, giving $\vert X_1 \vert \leq 4k \log \vert X \vert =\sqrt{\Delta \vert X \vert \log \vert X \vert} / 4$.

    In general, to determine $X_i$, let $X_1, X_2, \dots, X_{i - 1}$ be the sets found by iteratively applying Lemma \ref{lem:small_set_big_impact} to $X \setminus X_1$, $\dots$, $X \setminus \bigcup_{j \in [i - 2]} X_j$ respectively, and further observe that, by Lemma \ref{lem:small_set_big_impact}, each has size at most $\sqrt{\vert X \vert \Delta \log \vert X \vert} / 4$. As $\vert X_j \vert \leq \sqrt{\Delta \vert X \vert \log \vert X \vert} / 4$, we have that $\vert X \setminus \bigcup_{j \in [i - 1]} X_j\vert \geq \vert X \vert - i \sqrt{\Delta \vert X \vert \log \vert X \vert} / 4 \geq \vert X \vert/64$. Taking the lower bound of $\vert X \vert \geq 4m^2 + 4m + 4$ gives $\vert X \setminus \bigcup_{j \in [i - 1]} X_j\vert \geq 4m^2 + 4m + 4 - (m^2 + m + 1)/16 \geq m^2 + 4m + 4$ and thus we can apply Lemma \ref{lem:small_set_big_impact} to $X \setminus \bigcup_{j \in [i - 1]} X_j$ to find the set $X_i$, with $\vert X_i \vert \leq \sqrt{\vert X \vert \Delta \log \vert X \vert} / 4$.

    To determine the walk of length $\sqrt{\vert X \vert/(\Delta \log \vert X\vert)}$, we start by choosing some arbitrary vertex $v_{T/P + 1}\in X \setminus \bigcup_{i \in [T/P]} X_i$. Now, let $v_{T/P}$ be the vertex in $X_{T/P - 1}$ such that there exists a walk from $v_{T/P}$ to $v_{T/P + 1}$ in $\mathcal{G}^{T/P}$, noting that this vertex is guaranteed to exist by the above construction. More generally, let $v_i \in X_i$ be the vertex such that there exists a walk from $v_i$ to $v_{i + 1}$ in $\mathcal{G}^{i}$, observing that such a vertex exists for every $i \in [T / P]$ by construction. As $v_{i} \neq v_j$, for any pair $i \neq j$, we can construct a temporal walk in $\mathcal{G}$ visiting each vertex, starting at $v_1$ and ending at $v_{T/P + 1}$, and thus visiting $\sqrt{\vert X \vert / \Delta \log \vert X \vert} / 16$ vertices in $\vert X \vert$.
\end{proof}

\begin{lemma}
    \label{lem:spliting_X_by_go_and_return}
    Let $\mathcal{G} = G_1, G_2, \dots, G_T$ be an always $S$-connected temporal graph with lifetime $T = 2n m^2 + nm$. Then, given any subset $X \subseteq V$, there exists a subset $X'$ of $X$ of size at least $\vert X \vert / 2m^2$ such that there exists some $t \in [2m + 1]$, and vertex $v \in S$ such that, for every $u \in X'$, there exists a path from $v$ to $u$ in $G_{1}, \dots, G_{2t(T/(2m + 1))}$, and from $u$ to $v$ in $G_{(t + 1)T/(2m + 1) + 1}, \dots, G_{T}$.
\end{lemma}

\begin{proof}
    First, note that, for any sequence of at least $m \cdot n$ snapshots and vertex $v \in S$, there must exist at least one $u \in X$ such that there exists some path from $v$ to $u$ and $u$ to $v$. Let $R(u, t)$ denote the set of vertices in $S$ such that there exists a path from $v$ to $u$ and $u$ to $v$ for every $u \in R(u, t)$ in the graph $G_{(t - 1)T/(2m + 1) + 1}, \dots, G_{tT/(2m + 1)}$. Observe that by a pigeonhole argument, there must be at least one pair $t_1, t_2 \in [2m + 1]$ where $t_1 + 1 < t_2$ such that $R(u, t_1) \cap R(u, t_2) \neq \emptyset$. Now, let $R'(u, t) = \left(\bigcup_{t_1 \in [1, t - 1]} R(u, t_1)\right) \cap \left(\bigcup_{t_2 \in [t + 1, 2m+1]} R(u, t_2)\right)$, noting that, by the previous argument, for every $u \in X$ there must be at least one value of $t \in [2, 2m]$ for which $R'(u, t) \neq \emptyset$. More broadly, there must, therefore, be at least one $t \in [2, 2m]$ and $v \in S$ such that $\vert \{u \in X \mid v \in R'(u, t) \}\vert \geq \vert X \vert /2m^2$, giving the result.
\end{proof}

We can now prove our main result. At a high level, we combine Lemmas \ref{lem:path_length} and \ref{lem:spliting_X_by_go_and_return} to explore the set of unvisited vertices in an iterative manner, exploring $\sqrt{\vert X' \vert / (\Delta \log \vert X' \vert)} / 16m$ in each \emph{epoch}, while maintaining the key property that, after each exploration, we have exactly one agent on each vertex in $S$. This way, we allow ourselves to repeatedly apply the same approach in each epoch to an increasingly small number of vertices, until we are left with some set of $O(m^2)$, which we may explore via Lemma \ref{lem:there_and_back_again} in $O(n m^3)$ snapshots.

\begin{theorem_exploring_S_always_connected}
    Let $\mathcal{G} = G_1, G_2, \dots, G_T$ be an always $S$-connected temporal graph with lifetime $$T = \Omega\left(
    n m^2 \Delta^{3/2} \sqrt{\vert X \vert \log\left(m\vert X \vert\right)}\log(\vert X \vert) + nm^3\right)$$ where $m = \vert S \vert$ and $\Delta$ is the average degree of the underlying graph. Then any subset $X \subseteq V$ of vertices can be explored by $m$ agents.
\end{theorem_exploring_S_always_connected}

\begin{proof}
    We achieve our exploration by using a set of epochs, each of length $P = 4\Delta(2m + 1)n$.

    First, using Lemma \ref{lem:spliting_X_by_go_and_return}, we find some subset $X'$ of size $\vert X \vert /2m^2$ as described in the statement of the Lemma for the graph $G_1, \dots, G_P$, with the vertex $v \in S$ able to reach each vertex $u \in X'$ in the subgraph $G_1, \dots, G_{(t - 1)P/(2m + 1)}$ and each $u \in X'$ able to return to $v$ in the graph $G_{(t + 1)P/(2m + 1)}$, for some $t \in [2, 2m + 1]$.
    From Lemma \ref{lem:path_length}, we can find a temporal walk in the subsequence of snapshots $G_{(t - 1)P/(2m - 1) + 1}, \dots, G_{tP(2m - 1)}$ covering the subset $\chi \subseteq X'$ of size at least $\sqrt{\vert X' \vert / (\Delta \log\vert X' \vert)} / 16$ vertices in $X'$, and thus $\sqrt{\vert X \vert / \left(\Delta \log(\vert X \vert / 2m^2)\right)} / 32m$ of the vertices in $X$. Therefore, within the sequence $G_1, \dots, G_{P}$ we can have an agent starting at vertex $v$, visit each vertex in $\chi$, and return to $v$. Note this maintains the key invariant that at the end of each epoch of $P$ snapshots, we have an agent on each vertex of $S$. We repeat this strategy until a subset of $X$ of size $O(m^2)$ remains unexplored. We visit these remaining vertices in $O(n m^3)$ snapshots using the direct approach of Lemma \ref{lem:there_and_back_again}.
    
    To determine the total number of snapshots needed, let $X_i$ be the set of unexplored vertices at the start of the $i^{th}$ epoch, with $X_1 = X$. Noting that during the $i^{th}$ epoch we can explore $\vert \sqrt{\vert X_i \vert / \left(\Delta \log(\vert X_i \vert / 2m^2)\right)} / 32m$ vertices of $X_i$, we have $\vert X_{i + 1} \vert \leq$ $\vert X_{i} \vert -  \sqrt{\vert X_{i} \vert / \left(\Delta \log(\vert X_i \vert / 2m^2)\right)} / 32m$. Note that, after $\ell = 16m\sqrt{ \Delta \vert X \vert \left(\log(m\vert X \vert)\right)/2}$ epochs, we have $\vert X_{i + \ell} \vert \leq \max(\vert X_i \vert / 2, 2m^2)$. Thus, after $O(\ell \log \vert X \vert)$ epochs we have $\vert X_{\ell \log \vert X \vert}\vert = O(m^2)$. Finally, by Lemma \ref{lem:there_and_back_again}, we can visit the remaining vertices in $O(n m^3)$ snapshots. Thus, if $T = \Omega(\ell P\log \vert X \vert + n m^3) = \Omega\left(
    n m^2 \Delta^{3/2} \sqrt{\vert X \vert \log\left(m\vert X \vert\right)}\log(\vert X \vert) + nm^3\right)$
    snapshots we can explore all vertices in $X$ with $m$ agents.\looseness=-1
\end{proof}

\begin{corollary}
    \label{col:total_exploration_m_agents}
    Let $\mathcal{G} = G_1, G_2, \dots, G_T$ be an always $S$-connected temporal graph with lifetime 
    $$T = \Omega\left(
    n^{3/2} m^2 \Delta^{3/2}\log^{3/2}(n) + nm^3\right).$$ Then all vertices in $G$ can be explored by $m$ agents, with each starting on one of the vertices in $S$.
\end{corollary}

\section{Bounded Treewidth}
\label{sec:applications}

We now adapt the results from \cite{Adamson2022} to provide improved bounds on the exploration of graphs of bounded treewidth. First, we provide a bound on the number of snapshots needed to explore always-connected temporal graphs whose underlying graph contains an \emph{$(r,b)$-division}.

\begin{definition}[$(r,b)$-Division of Graphs]
    A set of vertices $S \subseteq V$ is an \emph{$(r,b)$-Division} of the graph $G = (V, E)$ if there is set of components covering the subgraph $G[V \setminus S]$ such that each component satisfies the following:
    \begin{itemize}
        \item The size of each component is at most $r$.
        \item There are at most $b$ vertices in $S$ that are incident to any vertex in the component, called the \emph{boundary vertices}.
        \item Given any edge $(v, u) \in E$, either both vertices are in the same component, both vertices are in $S$ or one vertex is in a component, and one is in $S$.
    \end{itemize}

    \noindent
    An $(r,b)$-division is \emph{strict} if there are $O(n/r)$ components of size $r$ in $G[V \setminus S]$.
\end{definition}

\begin{theorem}
    \label{thm:R_B_devision}
    Any always-connected temporal graph whose underlying graph has a strict $(r,b)$-division can be explored in
    $O\left(
    nr^{1/2} b^2 \Delta^{3/2}\log^{3/2}(r) + n^2/r + nb^3\right)$ snapshots  by $b$ agents.
\end{theorem}

\begin{proof}
    We can explore each component of size $(r + b)$ in $O(
    r^{3/2} b^2 \Delta^{3/2}\log^{3/2}(r) +$ $rb^3)$ snapshots per Theorem \ref{thm:exploring_S_always_connected}. As there are $O(n / r)$ such components, we can explore all components in $O\left(nr^{1/2} b^2 \Delta^{3/2}\log^{3/2}(r) + nb^3\right)$ snapshots. Thus, as we require at most $n$ snapshots to move the set of agents between components, we need an additional $O(n^2/r)$ snapshots, giving the total number of snapshots needed as $O\left(nr^{1/2} b^2 \Delta^{3/2}\log^{3/2}(r) + n^2/r + nb^3\right)$.
\end{proof}

\begin{theorem}
    \label{thm:r_b_one_agent}
    Any always-connected temporal graph whose underlying graph has a strict $(r,b)$-division can be explored by a single agent in $$O\left(nr^{1/2} b^3 \Delta^{3/2}\log^{3/2}(r)\log(n) + n^2\log(n)/r + nb^4\log n\right)$$ snapshots.
\end{theorem}

\begin{proof}
    This follows by applying Lemma \ref{lem:multi_agent_to_single_agent} to Theorem \ref{thm:R_B_devision}.
\end{proof}


\begin{lemma}[Generalisation of Lemma 4.4 \cite{erlebach2021temporal}]
    Every graph of treewidth at most $k$ admits a strict $(n^{2/3}, 6k)$-division.
\end{lemma}

\begin{theorem_treewidth_exploration}
    Any always-connected temporal graph whose underlying graph has treewidth $k$ can be explored by a single agent in $O\left(n^{4/3} k^{9/2}\log^{5/2}(n)\right)$ snapshots.
\end{theorem_treewidth_exploration}

\section{Simple and Proper Temporal Graphs}

We now provide some results on the exploration of simple temporal graphs. Recall that a temporal graph is \emph{simple} if each edge is active in exactly one snapshot, and is \emph{proper} if no vertex is incident to more than a single temporal edge per snapshot. On the negative side, we prove that it is NP-hard to determine if a simple, proper temporal graph can be explored by a single agent. On the other hand, we provide linear-time algorithms for determining if an instance of \TEXP~can be solved on simple Paths and Cycles, alongside an $O(n)$ time algorithm to determine the minimum number of agents needed to explore a given simple temporal path (assuming we can choose any vertices as the starting positions), and an $O(n^2)$ time algorithm to determine the minimum number of agents needed to explore a given simple temporal cycle (under the same conditions). Finally, we show that \TEXP~for a single agent can be solved on $k$-Chorded graphs in $O(n + (k^2 + k)!)$ time.

\subsection{Hardness}

First, we show that this problem is NP-hard via a reduction from the Hamiltonian path problem with a given starting vertex.
At a high level, our reduction works by taking an instance of the Hamiltonian path problem on the graph $G = (V, E)$, then constructing an auxiliary set of vertices and edges, forcing the only possible exploration to correspond to a Hamiltonian path.
We define our temporal graph as $\mathcal{G} = G_1, G_2, \dots, G_{2n + \vert E \vert}$ over the set of vertices $\mathcal{V}$.
We first show the construction for simple temporal graphs, before generalising to proper temporal graphs.


\noindent
\textbf{Vertices}. Our vertex set $\mathcal{V}$ contains the set $V$ of vertices from the Hamiltonian cycle instance, and an auxiliary set $B = \{v_{i,j,p} \mid p \in [n], (v_i, v_j) \in E \} \cup \{ b_0\}$, giving $\mathcal{V} = V \cup B$.
At a high level, the set $B$ is used to represent the possible edges in $E$, with $b_0$ used as an auxiliary vertex for a final ``clean up'' step.


\noindent
\textbf{Edges}.
Our edge set is constructed in a two-step process, the first corresponding to the first $2n - 1$ snapshots used to represent the underlying graph, and the second corresponding to the next $\vert B \vert$ snapshots, representing a ``clean-up '' phase, used to explore the vertices in the set $B$.

For each edge $(v_i, v_j) \in E$ and $p \in [n]$, we construct the edges $(v, v_{i,j,p})$ and $(v_{i,j,p}, v_j)$ and add $(v, v_{i,j,p})$ to the snapshot $G_{2p - 1}$ and $(v_{i,j,p}, v_j)$ to $G_{2p}$. These edges are used to represent the edges within $G$, with the pair $(v_{i}, v_{i, j, p})$ and $(v_{i, j, p}, v_j)$ representing moving from $v_i$ to $v_j$ as the $p^{th}$ step in the exploration.

Secondly, we add a set of edges between the edges in $B$. We add to snapshot $E_{2n}$ the edge $(v_i, b_0)$, for every $i \in [n]$. Then, letting $b_1, b_2, \dots, b_{\vert B \vert - 1}$ be some arbitrary ordering of the vertices in $B \setminus \{b_0\}$, we add the edge $(b_{i - 1}, b_{i})$ to $E_{2n + i}$, for every $i \in [1, \vert B \vert]$.

\noindent
\textbf{Starting Vertex}.
We assume, without loss of generality, that the agent starts on vertex $v_1$, and that the input instance asks if there is a Hamiltonian path in the input graph starting at $v_1$. Note that this problem is trivially NP-hard, and as such we use this to simplify our construction by removing any consideration of changing the exact starting location of the agent.



\begin{theorem}
    It is NP-hard to determine if a simple temporal graph can be explored by a single agent.
\end{theorem}

\begin{proof}
    First, let $(v_1, v_2), (v_2, v_3), \dots, (v_{n - 1}, v_n)$ be a Hamiltonian path in $G$. We construct an exploration of $\mathcal{G}$ as follows. Let $v_1$ be the starting vertex. For each $i \in [n - 1]$, we add, to our temporal walk, the edge $(v_i, v_{i, i + 1, i})$ at snapshot $2i - 1$, and the edge $(v_{i, i + 1, i}, v_{i + 1})$ at snapshot $2i$.
    At time $2n$, we traverse the edge $(v_n, b_0)$.
    Finally, for every $b_i \in B$, we traverse the edge $(b_{i - 1}, b_{i})$ at time $2n + i$.
    To show correctness, observe that, by construction, the edge $(v_i, v_{i, i + 1, i})$ is active in snapshot $2i - 1$, and the edge $(v_{i, i + 1, i}, v_{i + 1})$ is active in snapshot $2i$. Thus, this walk explores all vertices in $V$. Further, as the path $(b_0, b_1)$, $(b_1, b_2), \dots, (b_{\vert B \vert - 2}, b_{\vert B \vert - 1})$ exists in the snapshots $2n$ to $2n + \vert B \vert$ by construction, and visits all vertices of $B$, this walk visits all vertices in $\mathcal{V}$ starting at $v_1$.

    In the other direction, we show that if there exists any exploration of $\mathcal{G}$ starting at $v_1$, then there exists a Hamiltonian path in $G$ starting at $v_1$.
    First, note that, by construction, we can move from $v_i \in V$ to $v_j \in V$ between snapshots $2p - 1$ and $2p$ if and only if there is an edge between $v_i$ and $v_j$.
    Second, from snapshot $2n$ onwards, there is no possible path ending at any vertex in $V$. Therefore, the only way of exploring all vertices in $V$ is by way of a walk with $2n - 1$ edges in $\mathcal{G}$, corresponding to a Hamiltonian path requiring exactly $n$ edges in $G$.
    Finally, starting at any vertex $v_i$ in $V$, there exists a path from $v_i$ to all vertices in $B$, starting with the edge $(v_i, b_0)$ at snapshot $2n$, and following the path $(b_{j - 1}, b_{j})$ for every $b_j \in B$.
    Therefore, if there exists a walk exploring $\mathcal{G}$, then there exists a Hamiltonian path in $G$.
\end{proof}

We note that we can generalise this to proper temporal graphs by replacing each snapshot in $G_t$ with $\vert E_t \vert$ auxiliary snapshots, $G_t^1, G_t^2, \dots, G_t^{\vert E_t \vert}$, such that each snapshot contains exactly one edge from $E_t$, and the corresponding union of edge sets satisfies $E_t = \bigcup_{i \in [\vert E_t \vert]} E_t^i$.

\begin{corollary}
    It is NP-hard to determine if a simple, proper temporal graph can be explored by a single agent.
\end{corollary}

\subsection{Exploring Simple Temporal Paths and Cycles}
\label{subsec:paths_and_cycles}


We now consider the problems of exploring simple temporal graphs whose underlying graph is either a path or a cycle. We consider both the single-agent and multi-agent explorations for this setting, starting with determining the minimum number of agents needed to explore a given simple temporal path.

\begin{theorem}
    \label{thm:exploring_simple_paths}
    We can determine, in $O(n)$ time, the minimum number of agents needed to explore any simple temporal path $\mathcal{P}$.
\end{theorem}

\begin{proof}
    Let $v_1, v_2, \dots, v_n$ be the set of vertices of $\mathcal{P}$, with $v_{i, i + 1} \in E_t$, for some $t \in [T]$, for every $i \in [n - 1]$. 

    To solve this problem, we build two auxiliary tables, $\mathtt{TI}$ (standing for \textbf{T}able \textbf{I}ncreasing) and $\mathtt{TD}$ (standing for \textbf{T}able \textbf{D}ecreasing), where $\mathtt{TI}[i]$ contains the largest value $j$ such that $v_i$ can reach $v_j$ during the lifetime of the graph, and $\mathtt{TD}[i]$ the smallest value $k$ such that $v_i$ can reach $v_k$ during the lifetime of the graph.
    We compute $\mathtt{TI}$ starting at $i = 1$ and progressing by increasing value. We maintain two indices, $i$ and $j$, with $i$ and $j$ both initially equal to $1$. We increment the value of $j$ until we find the value such that $(v_{j}, v_{j + 1})$ is active before $(v_{j - 1}, v_j)$. We set $\mathtt{TI}[i'] = j$, for every $i' \in [i, j - 1]$, then set $i = j$ and repeat. We compute $\mathtt{TD}$ in an analogous manner, instead decreasing the values of $i$ and $j$. Note that both $\mathtt{TI}$ and $\mathtt{TD}$ take $O(n)$ time to compute,

    Using these tables, we determine the number of agents needed to explore $\mathcal{P}$ in a greedy manner. Observe that there must be at least one agent that either starts or ends on $v_1$. If that agent starts on $v_1$, then we assume, without loss of generality, that it ends of $\mathtt{TI}[1]$, and, equivalently, if it ends on $v_1$ we assume that it starts on $\mathtt{TD}[1]$.
    Now, let $j_2 = \max(\mathtt{TI}[1], \mathtt{TD}[1]) + 1$, and observe that to explore the subset of vertices $\{v_{j_2}, v_{{j_2} + 1}, \dots, v_{n}\}$ there must be some agent where either $v_{j_2}$ is the first vertex it explores, or $v_{j_2}$ the last. Thus, letting $j_3 = \max(\mathtt{TD}[j_2], \mathtt{TI}[j_2]) + 1$, we may assume that there is some agent either starting on $v_{j_3 - 1}$ and ending on $v_{j_2}$, or one starting on $v_{j_2}$ and ending on $v_{j_3 - 1}$.
    In general, assuming that $\kappa$ agents are needed to explore $\mathcal{P}$, let $j_k = \max(\mathtt{TI}[j_{k - 1}], \mathtt{TD}[j_{j - 1}]) + 1$, with the $k^{th}$ agent either starting on $v_{j_{k}}$ and ending on $v_{j_{k + 1} - 1}$ if $\mathtt{TI}[j_k] \geq \mathtt{TD}[j_k]$, or start on $v_{j_{k + 1} - 1}$ and end on $v_{j_k}$ if $\mathtt{TD}[j_k] > \mathtt{TI}[j_k]$, for every $k \in [\kappa]$, letting $j_{\kappa + 1} = n$ and $j_1 = 1$.

    To show correctness, observe that any exploration in which the agent visiting $v_1$ does not also visit $v_2, v_3, \dots, v_{j_2 - 1}$ can be replaced by an equivalent schedule in which the first agent either starts on $v_1$ and ends on $v_{j_2 - 1}$, or starts on $v_{j_2 - 1}$ and ends on $v_1$. Further, the exploration of the graph $\mathcal{P}'$, formed by removing all vertices in the set $v_1, v_2, \dots, v_{j_2 - 1}$ must, therefore, require exactly one fewer agent then the exploration of $\mathcal{P}$, and can be explored using the same initial exploration as $\mathcal{P}$, minus those vertices in the set $v_1, v_2, \dots, v_{j_2 - 1}$. Therefore, by induction, the exploration schedule given by the above greedy algorithm is optimal.

    As $\mathtt{TI}$ and $\mathtt{TD}$ can both be computed in $O(n)$ time, and the schedule derived from the tables in $O(\kappa)$ time, where $\kappa$ is the number of agents used in the exploration, we have a total complexity of $O(n + \kappa) = O(n)$, completing the proof.
\end{proof}

\begin{corollary}
    We can determine if an instance of \TEXP~on a simple temporal path can be solved in $O(n)$ time.
\end{corollary}

We now extend this same approach to cycles. At a high level, the key difference is that, to determine the minimum number of agents needed to explore each cycle, we check the $n$ possible combinations.

\begin{theorem}
    \label{thm:exploring_simple_cycles}
    We can determine in $O(n)$ time the number of agents needed to explore a simple temporal cycle.
\end{theorem}

\begin{proof}
    To determine the minimal number of agents, we employ a simple search-based approach, considering, for every $i \in [n]$, the minimum number of agents needed to explore the temporal path $v_i, v_{i + 1}, \dots,  v_{n}, v_1, \dots, v_{i - 1}$. Observing that there must be at least one unused edge $(v_i, v_{i + 1})$ in the optimal exploration schedule of the input simple temporal cycle, we have that the exploration of the path $v_i, v_{i + 1}, \dots,  v_{n}, v_1, \dots, v_{i - 1}$ must require the same number of agents as for the cycle, and, further, than for any $j \in [n] \setminus \{i\}$, the number of agents needed to explore the path $v_j, v_{j + 1}, \dots,  v_{n}, v_1, \dots, v_{j - 1}$ is the same or greater as to explore the cycle, we have correctness.
    The complexity follows from the $n$ possible induced paths, each requiring $O(n)$ time to determine the minimum number of agents to explore.
\end{proof}

\begin{corollary}
    \label{col:simple_cycle_exploration}
    We can determine if an instance of \TEXP~with $k$ agents on a simple temporal cycle can be solved in $O(n)$ time.
\end{corollary}

\begin{proof}
    Unlike the more general case, if given a fixed number of agents on specified starting vertices, we can instead simplify the computation of the existence of an exploration schedule into two cases. Assuming some agent starts on the vertex $v_1$, then either that agents traverses the edge $(v_1, v_2)$ or $(v_{1}, v_n)$. As such, we can construct two instances of \TEXP~in a simple temporal path, one of which is induced by removing the edge $(v_1, v_2)$, and the other by removing the edge $(v_1, v_n)$. If either instance is explorable, then so too is the cycle, and if neither can be explored, then the cycle can not, giving the complexity claim.
\end{proof}

\subsection{Simple Temporal $k$-Chorded Graphs}

Our final result concerns the exploration of simple temporal $k$-chorded graphs with a single agent. A graph $G = (V, E)$ with $n$ vertices is a $k$-chorded graph if $G$ contains $n + k$ edges such that there is some subset of $k$ edges $E' \subset E$ such that the graph $(V, E \setminus E')$ is exactly the $n$ vertex cycle graph. We do not place any additional restrictions on the chords.

At a high level, our approach is similar to that outlined in Section \ref{subsec:paths_and_cycles}. Letting our input graph be $\mathcal{G}$ with the underlying graph $U(\mathcal{G}) = (V, E)$, we fix some set of $k$ edges $E'$ as our chords, such that $(V, E \setminus E')$ forms a cycle, noting that there may be only one such set, or multiple. We assume, without loss of generality, that our vertices are labelled $v_1, v_2, \dots, v_n$ such that $(v_{i}, v_{i + 1}) \in E \setminus E'$, for every $i \in [n - 1]$, and $(v_n, v_1) \in E \setminus E'$.

The key observation is that we only have a choice in terms of our exploration on those vertices with a degree of at least three. To that end, let $V' = \{v \in V \mid \Delta(v) \geq 3\}$. 
We consider two cases based on the cardinality of $V'$. First, if $\vert V' \vert = 2$, and second, if $\vert V' \vert \geq 3$, noting that if $\vert V' \vert \leq 1$, then the original graph is a cycle, and thus \TEXP~ can be solved in $O(n)$ time.

\begin{lemma}
    \label{lem:exploring_with_one_chord_one_agent}
    We can determine if an instance of \TEXP~ with $1$ agent can be explored on a simple $1$-chorded temporal graph in $O(n)$ time.
\end{lemma}

\begin{proof}
    If the starting position of the agent is not on $V'$, then observe that the graph can only be explored if the temporal cycle with underlying graph $(V, E \setminus E')$ can be explored. Thus, using Corollary \ref{col:simple_cycle_exploration}, we can determine if this graph can be explored in $O(n)$ time.
    Otherwise, letting $v_1$ be the starting vertex with the chord $(v_1, v_i)$, we subdivide our possible solutions into three cases.

    First, the case where we do not use the chord $(v_1, v_i)$. In this case, we are exploring the temporal cycle with the underlying graph $(V, E \setminus E')$, which can be checked in $O(n)$ time.

    Second, the case in which we use the chord from $v_1$ ending at $v_i$. In this case, letting $t$ be the index where $(v_1, v_i) \in E_t$, we wait at $v_1$ until $G_t$, then cross the chord, leaving the remainder of the graph only explorable iff the temporal cycle on the graph $G_{t + 1}, G_{t + 2}, \dots, G_{T}$ can be explored, noting that we must cross $v_1$ in order to complete the exploration. Thus, the time cost to determine if there is some such exploration is equal to the cost of determining if $G_{t + 1}, G_{t + 2}, \dots, G_{T}$ can be explored, i.e. $O(n)$.

    Third, the case where we use the chord $(v_1, v_i)$ starting at $v_i$ and ending at $v_1$. In this case, we have two possible walks. First, the walk $(v_1, v_2), (v_2, v_3), \dots, (v_{i - 1}, v_i), (v_i, v_1), (v_{1}, v_{n}), \dots, (v_{j + 1}, v_j)$, and second the walk $(v_{1}, v_{n}), \dots, (v_{j + 1}, v_j), (v_i, v_1), (v_1, v_2), (v_2, v_3), \dots, (v_{i - 1}, v_i)$. As we can check each walk in $O(n)$ time, the total complexity is $O(n)$.
\end{proof}

We now consider the case where $\vert V' \vert \geq 3$. We assume that our agent starts on the vertex $v_1$, and further that $v_1$ belongs to the set $V'$, adding $v_1$ to $V'$ if it otherwise does not belong to it, and updating $E'$ to include any edges between $v_1$ and some other vertex in $V'$.
We construct an auxiliary temporal graph $\mathcal{G}'$ on the vertex set $W$ and underlying edge set $F$, using the sets $V'$ and $E'$ as our basis. Formally, let $W = V' \cup \{u_{i, j}, u_{i, j}' \mid v_i, v_j \in V'$, $i < j - 1$ and $\forall \ell \in [i + 1, j - 1]$, $v_\ell \notin V'\} \cup \{u_{j, 1} \mid 1 < j < n$, $\forall \ell \in [j + 1, n]$, $v_\ell \notin V'\}$, and $\mathcal{F} = E' \cup \{(v_i, u_{i,j}), (u_{i, j}, v_j) \mid v_i, v_j \in V', i \neq j \pm 1, u_{i, j} \in W\}$.

We now construct the timesteps. First, we add the edge $(v_i, v_j)$ to $F_t$ iff $(v_i, v_j) \in E_t$. Secondly, given a pair $v_i, v_j$ where $u_{i, j} \in W$, and time indices $t, t'$, we add the edge $(v_i, u_{i, j})$ to $F_t$ and $(u'_{i, j}, v_j)$ to $F_{t'}$ if one of the following holds:

\begin{itemize}
    \item $i < j$, $t < t'$, $(v_i, v_{i + 1}) \in E_t$, $(v_{j - 1}, v_{j}) \in E_{t'}$, and the edges $(v_{i}, v_{i + 1}), (v_{i + 1}, v_{i + 2}), \dots, (v_{j - 1}, v_j)$ appear in strictly increasing snapshots,
    \item $i < j$, $t > t'$, $(v_i, v_{i + 1}) \in E_t$, $(v_{j - 1}, v_{j}) \in E_{t'}$, and the edges $(v_{i}, v_{i + 1}), (v_{i + 1}, v_{i + 2}), \dots, (v_{j - 1}, v_j)$ appear in strictly decreasing snapshots,
    \item $j = 1$, $t > t'$, $(v_n, v_{1}) \in E_t$, $(v_{i + 1}, v_{i}) \in E_{t'}$, and the edges $(v_{1}, v_{n}), (v_{n}, v_{n - 1}), \dots, (v_{i + 1}, v_i)$ appear in strictly increasing snapshots, or, 
    \item $j = 1$, $t > t'$, $(v_1, v_{n}) \in E_t$, $(v_{i + 1}, v_{i}) \in E_{t'}$, and the edges $(v_{1}, v_{n}), (v_{n}, v_{n - 1}), \dots, (v_{i + 1}, v_i)$ appear in strictly decreasing snapshots.
\end{itemize}

\noindent
To account for the special case where we want to represent the exploration of every vertex between some pair $v_i$ and $v_j$ without ending on one or the other, we add the edge $(v_i, u_{i, j})$ to $F_t$ if:

\begin{itemize}
    \item $i < j$, $(v_i, v_{i + 1}) \in E_t$, and the edges $(v_{i}, v_{i + 1}), (v_{i + 1}, v_{i + 2}), \dots, (v_{j - 1}, v_j)$ appear in strictly increasing snapshots, or,
    \item $j = 1$, $(v_1, v_{n}) \in E_t$, and the edges $(v_{1}, v_{n}), (v_{n}, v_{n - 1}), \dots, (v_{i + 1}, v_i)$ appear in strictly decreasing snapshots.
\end{itemize}

\noindent
Alternatively, we add $(u'_{i, j}, v_j)$ to $F_{t}$ if:

\begin{itemize}
    \item $i < j$, $t$, $(v_{j - 1}, v_{j}) \in E_{t}$, and the edges $(v_{i}, v_{i + 1}), (v_{i + 1}, v_{i + 2}), \dots, (v_{j - 1}, v_j)$ appear in strictly decreasing snapshots, or,
    \item $j = 1$, $(v_{i + 1}, v_{i}) \in E_{t}$, and the edges $(v_{1}, v_{n}), (v_{n}, v_{n - 1}), \dots, (v_{i + 1}, v_i)$ appear in strictly increasing snapshots.
\end{itemize}

\noindent
Finally, for each $u_{i, j} \in W$, we add the edge $(u_{i, j}, u_{i.j}')$ to $F_t$ if one of the following holds:

\begin{itemize}
    \item $i < j$, $(v_{j - 2}, v_{j - 1}) \in E_{t}$, and the edges $(v_{i}, v_{i + 1}), (v_{i + 1}, v_{i + 2}), \dots, (v_{j - 1}, v_j)$ appear in strictly increasing snapshots,
    \item $i < j$, $(v_{i + 1}, v_{i + 2}) \in E_t$, and the edges $(v_{i}, v_{i + 1}), (v_{i + 1}, v_{i + 2}), \dots, (v_{j - 1}, v_j)$ appear in strictly decreasing snapshots,
    \item $j = 1$, $(v_{i + 2}, v_{i + 1}) \in E_{t}$, and the edges $(v_{1}, v_{n}), (v_{n}, v_{n - 1}), \dots, (v_{i + 1}, v_i)$ appear in strictly increasing snapshots, or, 
    \item $j = 1$, $(v_{n}, v_{n - 1}) \in E_t$, and the edges $(v_{1}, v_{n}), (v_{n}, v_{n - 1}), \dots, (v_{i + 1}, v_i)$ appear in strictly decreasing snapshots.
\end{itemize}

\noindent
For simplicity, if there exists only a single vertex between the pair $v_i, v_j \in V'$, then we merge $u_{i, j}$ and $u_{i, j}'$, noting that there can be no edge between them by the above construction. We will use both $u_{i, j}$ and $u_{i, j}'$ to refer to this vertex to avoid the need for duplicate explanations.

\begin{lemma}
    \label{lem:exploring_with_multiple_chords_one_agent_aux}
    The simple $k$-chorded temporal graph $\mathcal{G}$ can be explored by a single agent starting on the vertex $v_1$ if and only if the temporal graph $\mathcal{F} = (W, F_1), (W, F_2), \dots, (W, F_T)$ as constructed above can be explored.
\end{lemma}

\begin{proof}
    In one direction, given some exploration of $\mathcal{F}$ starting on $v_0$, visiting the vertices of $W$ in the order $w_1, w_2, \dots, w_{\ell}$, we can construct an exploration of $\mathcal{G}$ by:
    \begin{itemize}
        \item maintaining any edge $(w_i, w_{i + 1})$ taken at time $t$ where $w_i = v_i, w_{i + 1} = v_j$ for some $v_i, v_j \in V'$, noting that this edge must exist in both $\mathcal{G}$ and $\mathcal{F}$,
        \item replacing any walk of the form $(v_i, u_{i, j}), (u_{i, j}, u'_{i, j}), (u_{i, j}', v_j)$, traversed at times $t_1, t_2, t_3$ respectively, with the walk between $v_i$ and $v_j$ across either $(v_i, v_{i + 1})$, $(v_{i + 1}$, $v_{i + 2})$, $\dots$, $(v_{j - 1}, v_j)$, if $i < j$, or, $(v_i, v_{i - 1 \bmod n})$, $(v_{i - 1 \bmod n}, v_{i - 2\bmod n})$, $\dots$, $(v_{j + 1}, v_j)$, if $i < j$, traversing the first edge at time $t_1$ and the last at time $t_3$, observing that the existence of such a walk in $\mathcal{F}$ implies that the corresponding walk can be taken in $\mathcal{G}$.
    \end{itemize}

    \noindent
    Finally, if the last three vertices are of the form $v_i, u_{i, j}, u'_{i, j}$, then we replace the edges $(v_i, u_{i, j}), (u_{i,j}, u'_{i, j})$, traversed at times $t_1$ and $t_2$ respectively, with either $(v_i, v_{i + 1})$, $(v_{i + 1}, v_{i + 2})$, $\dots$, $(v_{j - 1}, v_j)$, if $i < j$, or, $(v_i, v_{i - 1 \bmod n})$, $(v_{i - 1 \bmod n}, v_{i - 2\bmod n})$, $\dots$, $(v_{j + 1}, v_j)$, if $i < j$, traversing the first edge at time $t_1$ and the last at time $t_2$, again noting that if $(v_i, u_{i, j}) \in E_{t_1}$, and $(u_{i,j}, u'_{i, j}) \in E_{t_2}$, then such a walk must exist in $\mathcal{G}$ by construction. Thus, if we can explore $\mathcal{F}$, then we can explore $\mathcal{G}$.

    In the other direction, if we can explore $\mathcal{G}$ by the walk $(v_0, v_{i_1}), (v_{i_1}, v_{i_2}), \dots, (v_{i_{\ell - 1}} v_{i_{\ell}})$, then we construct a new walk in $\mathcal{F}$ as follows. Let $v_{j_1}, v_{j_2}, \dots, v_{j_m}$ be the set of vertices in $V'$ in the order they are visited by our walk, allowing the same vertex to appear multiple times if visited multiple times.
    We construct our walk in $\mathcal{F}$ as follows.
    For each pair $j_{a}, j_{a + 1}$, where $v_{j_a} = v_b$ and $v_{j_{a + 1}} = v_c$ we replace the path $(v_{j_a}, v_{b + 1}), (v_{b + 1}, v_{b + 2}), \dots, (v_{c - 1}, v_{j_b})$ starting at time $t$ and ending at time $t'$ with:
    \begin{itemize}
        \item $(v_{j_a}, v_{j_{a + 1}})$ if $v_{j_{a + 1}}$ follows $v_{j_a}$ directly in the walk, observing that by construction this edge must exist, or,
        \item $(v_{j_a}, u_{j_a, j_{a + 1}})$ at time $t$, $(u_{j_a, j_{a + 1}}, u_{j_a, j_{a + 1}}')$ at time $t'' \in [t + 1, t' - 1]$ where $(v_{c - 2}, v_{c - 1}) \in E_{t''}$, and $(u_{j_a, j_{a + 1}}', v_{j_{a + 1}})$ at time $t'$. Observe that all three edges must exist in $\mathcal{F}$ by construction. 
    \end{itemize}

    \noindent
    Note that in the case that $v_{i_\ell}$ does not belong to $V'$, then it must be the neighbour of some vertex $v_j \in V'$. Thus, letting $v_{j'}$ be the last vertex in $V'$ that is visited in this exploration, we add to the walk exploring $\mathcal{F}$ $(v_{j'}, u_{j, j'})$ at time $t$ where $(v_{j'}, v_{i_{a}}) \in E_t$ for the vertex $i_a$ following the last appearance of $v_j'$ in the exploration of $\mathcal{G}$, and the edge $(u_{j, j'}, u_{j, j'}')$ at time $T'$, where $T'$ is the last snapshot used in the exploration.
    As this walk explores $\mathcal{G}$, it must also visit every vertex in $V'$ by definition, and each vertex in $W \setminus V'$, corresponding to the exploration of those vertices in $V \setminus V'$. Hence we have the complete reduction.
\end{proof}



\begin{theorem}
    We can determine in $O(n + (k^2 + k)!)$ time if a simple, proper temporal $k$-chordal graph $\mathcal{G}$ can be explored by a single agent.
\end{theorem}

\begin{proof}
    We prove this in a two-step process. First, let $\mathcal{F}$ be the auxiliary graph constructed above, observing that this graph requires $O(n + k)$ time to construct.
    As $\mathcal{F}$ contains at most $k^2 + k$ edges, we can determine if there is an exploration of $\mathcal{F}$ in $O((k^2 + k)!)$ time by a brute force approach. Following Lemma \ref{lem:exploring_with_multiple_chords_one_agent_aux}, we can transform the exploration of $\mathcal{F}$ in $O(n)$ time to an exploration of $\mathcal{G}$. Thus, we have the claim.
\end{proof}

\section{Acknowledgements}

The authors would like to thank the comments from the reviewers of the conference version of this paper, which aided greatly in improving the readability. Further, the results on the exploration of simple temporal graphs were devised following a talk at Dagstuhl Seminar 26251 (Temporal Graphs: Structure, Algorithms, Applications).



\printbibliography

@InProceedings{Adamson2022,
  author       = {Adamson, Duncan and Gusev, Vladimir V and Malyshev, Dmitriy and Zamaraev, Viktor},
  booktitle    = {1st Symposium on Algorithmic Foundations of Dynamic Networks (SAND 2022)},
  title        = {Faster exploration of some temporal graphs},
  year         = {2022},
  organization = {Schloss Dagstuhl--Leibniz-Zentrum f{\"u}r Informatik},
  pages        = {5--1},
}

@phdthesis{taghian2020exploring,
  title={Exploring temporal cycles and grids},
  author={Taghian Alamouti, Shadi},
  year={2020},
  school={Concordia University}
}

@Article{Akrida2021,
  author    = {Akrida, Eleni C and Mertzios, George B and Spirakis, Paul G and Raptopoulos, Christoforos},
  journal   = {Journal of Computer and System Sciences},
  title     = {The temporal explorer who returns to the base},
  year      = {2021},
  pages     = {179--193},
  volume    = {120},
  publisher = {Elsevier},
}

@inproceedings{meeks2022reducing,
  title={Reducing Reachability in Temporal Graphs: Towards a More Realistic Model of Real-World Spreading Processes},
  author={Meeks, Kitty},
  booktitle={Conference on Computability in Europe},
  pages={186--195},
  year={2022},
  organization={Springer}
}

@inproceedings{dogeas2023exploiting,
  author       = {Konstantinos Dogeas and
                  Thomas Erlebach and
                  Frank Kammer and
                  Johannes Meintrup and
                  William K. Moses Jr.},
  title        = {Exploiting Automorphisms of Temporal Graphs for Fast Exploration and
                  Rendezvous},
  booktitle    = {51st International Colloquium on Automata, Languages, and Programming,
                  {ICALP} 2024, July 8-12, 2024, Tallinn, Estonia},
  series       = {LIPIcs},
  volume       = {297},
  pages        = {55:1--55:18},
  year         = {2024}
}

@InProceedings{Arrighi2023,
  author       = {Arrighi, Emmanuel and Fomin, Fedor V and Golovach, Petr A and Wolf, Petra},
  booktitle    = {18th International Symposium on Parameterized and Exact Computation (IPEC 2023)},
  title        = {Kernelizing Temporal Exploration Problems},
  year         = {2023},
  organization = {Schloss Dagstuhl--Leibniz-Zentrum f{\"u}r Informatik},
  pages        = {1--1},
}

@article{michail2016traveling,
  title={Traveling salesman problems in temporal graphs},
  author={Michail, Othon and Spirakis, Paul G},
  journal={Theoretical Computer Science},
  volume={634},
  pages={1--23},
  year={2016},
  publisher={Elsevier}
}

@Misc{Baguley2025,
  author        = {Samuel Baguley and Andreas Göbel and Nicolas Klodt and George Skretas and John Sylvester and Viktor Zamaraev},
  title         = {Temporal Exploration of Random Spanning Tree Models},
  year          = {2025},
  archiveprefix = {arXiv},
  eprint        = {2508.03361},
  primaryclass  = {cs.DM},
  url           = {https://arxiv.org/abs/2508.03361},
}

@Article{Bastide2025,
  author  = {Bastide, Paul and Groenland, Carla and Michel, Lukas and Rambaud, Cl{\'e}ment},
  journal = {arXiv preprint arXiv:2511.22604},
  title   = {Improved exploration of temporal graphs},
  year    = {2025},
}

@Article{Enright2021,
  author    = {Enright, Jessica and Meeks, Kitty and Mertzios, George B and Zamaraev, Viktor},
  journal   = {Journal of Computer and System Sciences},
  title     = {Deleting edges to restrict the size of an epidemic in temporal networks},
  year      = {2021},
  pages     = {60--77},
  volume    = {119},
  publisher = {Elsevier},
}

@Article{erlebach2022exploration,
  author    = {Erlebach, Thomas and Spooner, Jakob T},
  journal   = {Acta Informatica},
  title     = {Exploration of k-edge-deficient temporal graphs},
  year      = {2022},
  number    = {4},
  pages     = {387--407},
  volume    = {59},
  publisher = {Springer},
}

@Article{erlebach2021temporal,
  author    = {Erlebach, Thomas and Hoffmann, Michael and Kammer, Frank},
  journal   = {Journal of Computer and System Sciences},
  title     = {On temporal graph exploration},
  year      = {2021},
  pages     = {1--18},
  volume    = {119},
  publisher = {Elsevier},
}

@inproceedings{erlebach2019two,
  title={Two moves per time step make a difference},
  author={Erlebach, Thomas and Kammer, Frank and Luo, Kelin and Sajenko, Andrej and Spooner, Jakob T},
  booktitle={46th International Colloquium on Automata, Languages, and Programming (ICALP 2019)},
  year={2019}
}

@article{erlebach2023parameterised,
  title={Parameterised temporal exploration problems},
  author={Erlebach, Thomas and Spooner, Jakob T},
  journal={Journal of Computer and System Sciences},
  volume={135},
  pages={73--88},
  year={2023},
  publisher={Elsevier}
}

@article{bodlaender2019exploring,
  title={On exploring always-connected temporal graphs of small pathwidth},
  author={Bodlaender, Hans L and van der Zanden, Tom C},
  journal={Information Processing Letters},
  volume={142},
  pages={68--71},
  year={2019},
  publisher={Elsevier}
}

@Article{Kutner2025a,
  author  = {Kutner, David C and Sommer, Anouk},
  journal = {arXiv preprint arXiv:2501.18987},
  title   = {Better late, then? The hardness of choosing delays to meet passenger demands in temporal graphs},
  year    = {2025},
}

@article{DeligkasP22,
  author    = {Argyrios Deligkas and
               Igor Potapov},
  title     = {Optimizing reachability sets in temporal graphs by delaying},
  journal   = {Inf. Comput.},
  volume    = {285},
  number    = {Part},
  pages     = {104890},
  year      = {2022}
}

@article{ruget2021multi,
  title={Multi-species temporal network of livestock movements for disease spread},
  author={Ruget, Anne-Sophie and Rossi, Gianluigi and Pepler, P Theo and Beaun{\'e}e, Ga{\"e}l and Banks, Christopher J and Enright, Jessica and Kao, Rowland R},
  journal={Applied Network Science},
  volume={6},
  pages={1--20},
  year={2021},
  publisher={Springer}
}

\end{document}